\documentclass{article}
\usepackage{graphicx} % Required for inserting images

\usepackage{amsfonts}
\usepackage{amssymb}
\usepackage{amstext}
\usepackage{amsmath}
\usepackage{enumerate}
\usepackage{algorithm}
\usepackage{algorithmic}

\usepackage{graphicx}

\usepackage{bbm}

\usepackage{amsthm}

\usepackage{thmtools}

\usepackage{verbatim} % comment environment

\usepackage[dvipsnames,svgnames,table]{xcolor}
\definecolor{dartmouthgreen}{rgb}{0.05, 0.5, 0.06}
\definecolor{ceruleanblue}{rgb}{0.16, 0.32, 0.75}

\usepackage[colorlinks=true,pdfpagemode=UseNone,citecolor=dartmouthgreen,linkcolor=ceruleanblue,urlcolor=BrickRed,pdfstartview=FitH]{hyperref}
\usepackage{cleveref}

\usepackage{framed}
\usepackage{enumitem}


\newtheorem{theorem}{Theorem}[section]

\newtheorem{lemma}[theorem]{Lemma}
\newtheorem{definition}[theorem]{Definition}

\newtheorem{claim}[theorem]{Claim}

\newtheorem*{problem*}{Problem}

\newtheorem*{remark*}{Remark}

\numberwithin{equation}{section}
\numberwithin{table}{section}

\newcommand{\R}{\ensuremath{\mathbb R}}

\newcommand{\E}[1]{{\mathbb{E}}\left[#1\right]}

\newcommand{\poly}{\operatorname{poly}}

\newcommand{\junk}[1]{}

\newcommand{\vertiii}[1]{{\left\vert\kern-0.25ex\left\vert\kern-0.25ex\left\vert #1 \right\vert\kern-0.25ex\right\vert\kern-0.25ex\right\vert}}

\DeclareMathAlphabet{\mathsfit}{T1}{\sfdefault}{\mddefault}{\sldefault}
\SetMathAlphabet{\mathsfit}{bold}{T1}{\sfdefault}{\bfdefault}{\sldefault}

\def\vu{{\mathsfit u}}
\def\vv{{\mathsfit v}}
\def\vw{{\mathsfit w}}
\def\vx{{\mathsfit x}}

\def\mA{{\mathsfit A}}

\def\mI{{\mathsfit I}}

\def\mM{{\mathsfit M}}

\def\mX{{\mathsfit X}}

\def\b1{{\bf 1}}
\def\eps{{\varepsilon}}

\def\R{\mathbb{R}}

\def\tr{\operatorname{tr}}

\global\long\def\E{\mathbb{E}}

\global\long\def\R{\mathbb{R}}

\DeclareMathOperator{\alphamin}{\alpha-min}

\title{Experimental Design Using Interlacing Polynomials}

\author{
Lap Chi Lau\footnote{Cheriton School of Computer Science, University of Waterloo, Canada. Email: \href{mailto:lapchi@uwaterloo.ca}{lapchi@uwaterloo.ca}.}
\and Robert Wang\footnote{Cheriton School of Computer Science, University of Waterloo, Canada. Email: \href{mailto:robert.wang2@uwaterloo.ca}{robert.wang2@uwaterloo.ca}}
\and Hong Zhou\footnote{School of Mathematics and Statistics, Fuzhou University, China. Email: \href{mailto:hong.zhou@fzu.edu.cn}{hong.zhou@fzu.edu.cn}.}}

\date{}

\begin{document}

\maketitle

\begin{abstract}
We present a unified deterministic approach for experimental design problems using the method of interlacing polynomials. Our framework recovers the best-known approximation guarantees for the well-studied D/A/E-design problems with simple analysis. Furthermore, we obtain improved non-trivial approximation guarantee for E-design in the challenging small budget regime. Additionally, our approach provides an optimal approximation guarantee for a generalized ratio objective that generalizes both D-design and A-design.
\end{abstract}

\newpage

\section{Introduction}

Experimental design is a classical problem in statistics~\cite{Puk06}, which recently found wide applications from machine learning (e.g., active learning, feature selection, data summarization) to numerical linear algebra (e.g., column subset selection, sparse least squares regression) to graph algorithms (e.g., total effective resistance minimization, algebraic connectivity maximization).
We refer the reader to~\cite{SX20,AZLSW21,NST22,LZ22a,LZ22b,LWZ23} and the references therein for additional background and related applications.

In experimental design problems, we are given vectors $\vv_1,\ldots,\vv_m \in \R^d$ and a budget $k \geq d$ as input,
and the goal is to choose a multiset $S$ of $k$ vectors so that $\sum_{i \in S} \vv_i \vv_i^\top$ optimizes some objective function that measures the ``diversity'' of a solution.
The most popular and well-studied experimental design objective functions are:
\vspace{-5pt}
\begin{itemize}
\setlength\itemsep{-1pt}
\item D-design: Maximizing $\left(\det\left( \sum_{i \in S} \vv_i \vv_i^\top \right)\right)^{\frac{1}{d}}$,
\item A-design: Minimizing $\tr\left( \left(\sum_{i \in S} \vv_i \vv_i^\top \right)^{-1} \right)$,
\item E-design: Maximizing $\lambda_{\min}\left(\sum_{i \in S} \vv_i \vv_i^\top\right)$.
\end{itemize}
\vspace{-5pt}

All three problems are computationally hard to solve exactly~\cite{SEFM15,NST22,CM09}.
The following convex programming relaxations\footnote{For consistency, we present all D/A/E-design as minimization problems and thus the approximation ratios in the later discussions will all be greater than one.} have been widely used in designing efficient approximation algorithms for these problems~\cite{SX20,NST22,AZLSW21}, where we denote $\vx(i)$ as the $i$-th entry of the vector $\vx$.
\begin{equation} \label{eq:CP}
    \begin{aligned}
        & \underset{\vx \in \R^m}{\rm minimize} & & \det\bigg( \sum_{i=1}^m \vx(i) \vv_i \vv_i^\top \bigg)^{-\frac{1}{d}} {\rm~~or~~} & & \tr\bigg( \sum_{i=1}^m \vx(i) \vv_i \vv_i^\top \bigg)^{-1} {\rm~~or~~} & & \left(\lambda_{\min}\bigg( \sum_{i=1}^m \vx(i) \vv_i \vv_i^\top \bigg) \right)^{-1} \\
        & \text{\rm subject to} & & \sum_{i=1}^m \vx(i) \leq k, \\
        & & & \vx(i) \geq 0, \quad ~\text{ for } i \in [m].
    \end{aligned} 
\end{equation}

\subsection{Previous Work}

Nikolov, Singh, and Tantipongpipat~\cite{NST22} designed a proportional volume sampling technique to provide a $k/(k-d+1)$-approximate rounding algorithm for D-design and A-design, which implies a $(1+\eps)$-approximation when $k \gtrsim d/\eps$ (i.e., $k \geq c d /\eps$ for some large constant $c$). 
These rounding results are optimal when $k$ is large in the sense that the approximation ratios match the integrality gaps of the relaxations in \eqref{eq:CP}.
The approximation ratio can be further improved to the constant $e$ for D-design when $k = d$~\cite{NST22}. 

Allen-Zhu, Li, Singh, and Wang~\cite{AZLSW17c,AZLSW21} developed a general regret minimization framework to provide a $(1+\eps)$-approximate rounding algorithm for all D/A/E-design when $k \gtrsim d/\eps^2$, which is optimal for E-design but not optimal for D/A-design.
Lau and Zhou~\cite{LZ22b} refined the regret minimization framework and designed a randomized local search method to unify the optimal guarantees for D/A/E-design.
The randomized local search algorithm is relatively easy to describe, but the analysis is quite involved and also cannot handle the regime where $k$ is close to $d$.

To summarize, there is no simple unifying approach that can recover optimal guarantees for all D/A/E-design for all $k \geq d$.
In particular, to the best of our knowledge, the regime of E-design when $k$ is close to $d$ (e.g., when $d \leq k \leq 2d$) has not yet been well-explored in the literature.

\subsection{Our Results}

The method of interlacing polynomials was introduced by Marcus, Spielman, and Srivastava~\cite{MSS15a,MSS15b,MSS22}, which led to the resolution of the long-standing Kadison-Singer problem~\cite{MSS15b}.
The starting point of this work is to note that the restricted invertibility problem studied in~\cite{MSS22} is essentially the E-design problem when $k < d$.
Our main contribution is to show that the method of interlacing polynomials, coupled with the convex programming relaxations in \eqref{eq:CP}, provides a simple deterministic framework for designing rounding algorithms for experimental design problems.

\begin{theorem} \label{t:E-design}
For any $k \geq d$, there is a $\big(1-\sqrt{(d-1)/k}\big)^{-2}$-approximation deterministic rounding algorithm for E-design using the method of interlacing polynomials. The rounding algorithm runs in $O(kmd^{\omega+1})$ time, where $\omega$ is the matrix-multiplication exponent.
\end{theorem}

\Cref{t:E-design} recovers the optimal approximation ratio in~\cite{AZLSW17c,AZLSW21} for E-design when $k$ is large, achieving $(1+\eps)$-approximation when $k \gtrsim d/\eps^2$. Moreover, it provides a nontrivial guarantee in the regime where $k$ is close to $d$.
In particular, when $k=d$, \autoref{t:E-design} implies a $d^2$-approximation algorithm for E-design. Note that a folklore result (which was not explicitly mentioned in the literature to our knowledge) is that any $c$-approximate solution for A-design is a $c d$-approximate solution for E-design. Thus, the A-design algorithm in \cite{NST22} is also a $kd/(k-d+1)$-approximation algorithm for E-design. However, this bound is suboptimal when $k$ is getting away from $d$. For example when $k=2d$, it gives a $O(d)$ approximation while \Cref{t:E-design} gives a constant approximation.

Furthermore, we made some interesting observations to use the method of interlacing polynomials for D/A-design, which allow us to recover the optimal approximation ratios in~\cite{NST22}  under a unified framework.

\begin{theorem} \label{t:D-design}
For any $k \geq d$, there is a $k \cdot ((k-d)!/k!)^{\frac1d}$-approximation deterministic rounding algorithm for D-design using the method of interlacing polynomials. The rounding algorithm runs in $O(kmd^{\omega+1})$ time.
\end{theorem}

Note that $k \cdot ((k-d)!/k!)^{1/d} \leq k/(k-d+1)$, which implies a $(1+\eps)$-approximation when $k \gtrsim d/\eps$.
When $k=d$, the above approximation guarantee is $d \cdot (1/d!)^{\frac1d} \leq e$, recovering the result in~\cite{NST22}.
We remark that, for D-design, our algorithm is essentially equivalent to the derandomized algorithm in Section 5 of~\cite{SX20}

\begin{theorem} \label{t:A-design}
For any $k \geq d$, there is a $k / (k-d+1)$-approximation deterministic rounding algorithm for A-design using the method of interlacing polynomials. The rounding algorithm runs in $O(kmd^{\omega+1})$ time.
\end{theorem}

In addition, the algorithmic framework can also be extended to the generalized ratio objective introduced in~\cite{NST22}, which generalizes both D-design and A-design objectives.
We refer the reader to \Cref{s:ratio} for more details about the definition of the generalized ratio objective and our results.

To summarize, the method of interlacing polynomials provides the first {\em deterministic} framework that unifies all the optimal guarantees for D/A/E-design as well as for the generalized ratio objective.
Comparing with the randomized local search approach in~\cite{LZ22b}, our framework provides considerably simpler analyses and also covers the whole regime of $k \geq d$.

\subsection{Technical Overview}

To motivate the method of interlacing polynomials for experimental design problems, we start with a failed probabilistic approach of solving E-design.
Let $\vv_1, \ldots, \vv_m \in \R^d$ be the input vectors and $\vx \in \R^m_+$ be an optimal fractional solution to \eqref{eq:CP} with $\sum_{i=1}^m \vx(i) = k$.
Suppose we sample $k$ i.i.d.\!~random vectors $\vu_1, \ldots, \vu_k$ from $\{\vv_1, \ldots, \vv_m\}$ with replacement, where $\vv_j$ is sampled with probability $\frac1k \vx(j)$ for all $j \in [m]$ in each round.
The classical probabilistic method asserts that there exists a solution $(s_1, \ldots, s_k) \in [m]^k$ such that
\[
\lambda_{\min}\left( \sum_{i=1}^k \vv_{s_i} \vv_{s_i}^\top \right)  \geq \E\left[ \lambda_{\min}\left( \sum_{i=1}^k \vu_i \vu_i^\top \right) \right].
\]
To find such a solution $(s_1, \ldots, s_k)$ efficiently, the conditional expectation method fixes an $s_i \in [m]$ in each iteration so that 
\[
\E\bigg[ \lambda_{\min}\bigg(\sum_{j=1}^k \vu_j \vu_j^\top \bigg) \left| \vu_1 = \vv_{s_1}, \cdots, \vu_{i} = \vv_{s_{i}} \bigg] \right. \geq \E \bigg[ \lambda_{\min}\bigg(\sum_{j=1}^k \vu_j \vu_j^\top \bigg) \left| \vu_1 = \vv_{s_1}, \cdots, \vu_{i-1} = \vv_{s_{i-1}} \bigg] \right. .
\]
However, this approach fails for two reasons: (1) the bound $\E[ \lambda_{\min}( \sum_{i=1}^k \vu_i \vu_i^\top ) ]$ could be too small comparing to the optimal value of the convex programming relaxation in \eqref{eq:CP}\footnote{For example, suppose the input contains $d$ copies of a standard basis and $k=d$. In this case, the optimal fractional solution is $\vx(i) = 1/d$ for all $i \in [m]$ and the optimal value of \eqref{eq:CP} is 1. However, $\E[ \lambda_{\min}( \sum_{i=1}^k \vu_i \vu_i^\top ) ] = 1/d^d$ tends to 0 as $d \to \infty$. Thus, this basic probabilistic method does not provide a bounded approximation ratio. \label{fn:eg}}, and (2) we do not know how to efficiently compute all conditional expectations.

The method of interlacing polynomials can be treated as a more powerful variant of the classical probabilistic method. 
We present an informal discussion here and refer the reader to \autoref{s:prelim} for formal definitions.
To address the first issue of a small expected minimum eigenvalue, note that the minimum eigenvalue $\lambda_{\min}(\sum_{i=1}^k \vu_i \vu_i^\top)$ is exactly the minimum root of the characteristic polynomial $\det(x \mI_d - \sum_{i=1}^k \vu_i \vu_i^\top)$. 
Instead of considering the expected minimum eigenvalue $\E [\lambda_{\min} ( \det(x \mI_d - \sum_{i=1}^k \vu_i \vu_i^\top))]$, the method of interlacing polynomials considers the minimum root of the expected characteristic polynomial
\begin{align} \label{eq:char-poly}
\lambda_{\min}\left(\E\left[  \det \left( x \mI_d - \sum_{i=1}^k \vu_i \vu_i^\top  \right) \right] \right),
\end{align}
where the expectation is moved inside.
We note that this could be much larger than the expected minimum eigenvalue\footnote{In the example in \cref{fn:eg}, as what we will show later in \autoref{ss:E}, the minimum root of the expected polynomial is at least $(1-\sqrt{(d-1)/k})^2$, which is much larger than the expected minimum eigenvalue $1/d^d$.}. 
Indeed, an important point in our analysis is that the quantity \eqref{eq:char-poly} can be related to the optimal value of the convex relaxation \eqref{eq:CP}; see~\autoref{ss:E}.

As an analog to the classical probabilistic method, we would like to conclude that there exists a solution $(s_1, \ldots, s_k) \in [m]^k$ such that
\[
\lambda_{\min}\left( \sum_{i=1}^k \vv_{s_i} \vv_{s_i}^\top \right) = \lambda_{\min}\left( \det\left( x \mI_d - \sum_{i=1}^k \vv_{s_i} \vv_{s_i}^\top  \right) \right)  \geq \lambda_{\min}\left(\E\left[  \det \left( x \mI_d - \sum_{i=1}^k \vu_i \vu_i^\top  \right) \right] \right).
\]
Marcus, Spielman, and Srivastava~\cite{MSS22} developed an interesting method to prove it by organizing the $m^k$ real-rooted characteristic polynomials $\det(x \mI_d - \sum_{i=1}^k \vv_{s_i} \vv_{s_i}^\top)$ (one for each fixed $(s_1, \ldots, s_k) \in [m]^k$) into a hierarchical structure (or a rooted-tree) called {\em interlacing family} such that all the $m^k$ characteristic polynomials are the leaves of the tree and each internal node labelled by $(s_1, \ldots, s_i)$ is a real-rooted polynomial that is a convex combination of the children polynomials, which can be written in the form of a conditional expected characteristic polynomial
\begin{align} \label{eq:cond-char-poly}
\E\left[ \left. \det \left( x \mI_d - \sum_{j=1}^k \vu_j \vu_j^\top  \right)  \right| \vu_1 = \vv_{s_1}, \cdots, \vu_{i} = \vv_{s_{i}}  \right].
\end{align}
Moreover, the children polynomials of an internal node have a ``common interlacing'', which means that there exist {\em disjoint} intervals $I_j \subset \R$ such that the $j$-th largest roots of the children polynomials are all in $I_j$.
In other words, there are points on the real line that interlace/separate the $j$-th largest roots of the children polynomials for different $j$.
This property leads to the name ``interlacing family'', and it guarantees that the parent polynomial is also real-rooted~\cite[Theorem 2.10]{MSS22} and thus the minimum eigenvalue of \eqref{eq:cond-char-poly} is well-defined. 
More importantly, the ``interlacing'' property also ensures that the minimum root of the parent polynomial is ``sandwiched'' by the minimum roots of the children polynomials (see \autoref{t:interlacing}).
Thus, the interlacing family structure naturally leads to a greedy algorithm that mimics the conditional expectation method:
In each round $i$, we fix an $s_i \in [m]$ such that
\begin{multline*}
\lambda_{\min}\bigg( \E\bigg[ \det \bigg( x \mI_d - \sum_{j=1}^k \vu_j \vu_j^\top  \bigg)  \left| \vu_1 = \vv_{s_1}, \cdots, \vu_{i} = \vv_{s_{i}} \bigg] \bigg) \right. \\
\geq \lambda_{\min}\bigg( \E \bigg[ \det \bigg( x \mI_d - \sum_{j=1}^k \vu_j \vu_j^\top  \bigg) \left| \vu_1 = \vv_{s_1}, \cdots, \vu_{i-1} = \vv_{s_{i-1}} \bigg] \bigg) \right. ,
\end{multline*}
which can also be seen as following a path from the root to a leaf in the interlacing family.

Given any internal node $(s_1, \ldots, s_i)$, \cite{MSS22} provides an efficient way to compute the conditional expected characteristic polynomial \eqref{eq:cond-char-poly} and an approximate minimum root with high accuracy. This addresses the second issue about efficiently implementing the algorithm.

To extend the approach to D/A-design, note that the polynomial associated with a leaf $(s_1, \cdots, s_k)$ in the interlacing family can be written as
\[
p(x) = \prod_{i=1}^d (x - \lambda_i) = \det\left( x \mI_d - \sum_{j = 1}^k \vv_{s_j} \vv_{s_j}^\top \right),
\]
where $\lambda_1 \leq \cdots \leq \lambda_d$ are the eigenvalues of $\sum_{j=1}^k \vv_{s_j} \vv_{s_j}^\top$. 
We observe that
\begin{itemize}
    \item The constant coefficient of $p(x)$ is $p(0) = \prod_{j=1}^d (-\lambda_j)$ and thus $(-1)^d \cdot p(0) = \prod_{j=1}^d \lambda_j$ is exactly the D-design objective value of the solution $(s_1, \ldots, s_k)$;
    \item The linear term coefficient of $p(x)$ is $p'(0) = \sum_{i=1}^d \prod_{j \neq i} (-\lambda_j)$ and thus $- p'(0)/ p(0)$ is exactly the A-design objective value of the solution $(s_1, \ldots, s_k)$.
\end{itemize}

This notion of D/A-design ``objective value'' of a solution can be naturally extended to all internal polynomials in the interlacing family, i.e., conditional expected characteristic polynomials in \eqref{eq:cond-char-poly}.
Interestingly, we can show that the D/A-design objective value of an internal polynomial is also ``sandwiched'' by the objective values of the children polynomials. Therefore, the greedy algorithm for E-design that trickles down the interlacing family tree from the root can also be carried over to D/A-design.
We remark that the ``interlacing'' property of the interlacing family is not required to derive the ``sandwiching'' property we need for D/A-design objectives.
See \Cref{s:algorithms} for the formal description of the algorithms and the full analyses.

Finally, we note that the method of interlacing polynomials has been used for the restricted invertibility problem~\cite{MSS22} and the subset selection problem~\cite{XX21},
and some of the analyses in this paper were inspired by those in~\cite{MSS22,XX21}.
However, by incorporating the convex programming relaxation \eqref{eq:CP}, we obtain approximation guarantees that are much stronger than in~\cite{MSS22} and \cite{XX21}.

\subsection{Notations}

We use italic sans-serif font for vectors and matrices, e.g., $\vx$, $\mA$, and use subscripts to indicate different objects, e.g., $\vv_1, \vv_2, \dots$ for vectors and $\mA_1, \mA_2, \dots$ for matrices. 
To avoid confusion, we denote $\vv(i)$ as the $i$-th entry of a vector $\vv$. 
We write $\mI_d$ as the $d$-dimensional identity matrix.
We denote $\partial_x$ as the partial differentiation operator with respect to $x$. We write $f'$ as the derivative of a univariate function $f$, and write $f^{(k)}$ as the $k$-th derivative.
Given a real-rooted degree-$d$ polynomial $p(x) \in \R[x]$, we denote $\lambda_1(p) \leq \lambda_2(p) \leq \cdots \leq \lambda_d(p)$ as the roots of $p$, and denote $\lambda_{\min}(p) := \lambda_1(p)$ as the minimum root of $p$.

\section{The Interlacing Family} \label{s:prelim}

First we review some concepts from~\cite{MSS22}.
Let $p$ be a real-rooted degree $d+1$ polynomial with roots $\lambda_1 \leq \cdots \leq \lambda_{d+1}$ and $q$ be a real-rooted degree $d$ polynomial with roots $\mu_1 \leq \cdots \leq \mu_d$.
We say the roots of $q$ \textit{interlace} $p$ if $\lambda_1 \leq \mu_1 \leq \lambda_2 \leq \cdots \mu_d \leq \lambda_{d+1}$. In other words, if we drew the roots of $p$ and $q$ on the real line, they alternate. 
An interlacing family of polynomials is defined as follows.

\begin{definition}[{\cite[Definition 2.5]{MSS22}}] \label{d:family}
An {\em interlacing family} consists of a finite rooted tree $T$ and a labeling of the nodes $v \in T$ by monic real-rooted polynomials $f_v(x) \in \R[x]$ of the same degree. Furthermore, the polynomials satisfy the following two properties:
\begin{enumerate}
    \item Every polynomial $f_v(x)$ corresponding to a non-leaf node $v$ is a convex combination of the polynomials corresponding to the children of $v$.
    \item For all nodes $v_1, v_2 \in T$ with a common parent, all convex combinations of $f_{v_1}(x)$ and $f_{v_2}(x)$ are real-rooted.
\end{enumerate}
\end{definition}

The second property guarantees that the children polynomials of an internal node have a ``common interlacing'', such that there exists a polynomial $g$ that interlaces all the children polynomials (see Section 2 of~\cite{MSS22}). In particular, the interlacing property leads to an important ``sandwiching'' property of the roots of the polynomials that will be useful for bounding the roots of the children in the tree.

\begin{theorem}[{\cite[Theorem 2.7]{MSS22}}] \label{t:interlacing}
Suppose we are given an interlacing family of degree $d$ polynomials.
Let $v$ be an internal node of the interlacing family labeled with a polynomial $f_v$.
Then, there exist children $u$ and $u'$ of $v$ such that
\[
\lambda_{\min}(f_u) \geq \lambda_{\min}(f_v) \geq \lambda_{\min}(f_{u'}).
\]
This implies that there exists a leaf $u$ such that $\lambda_{\min}(f_u) \geq \lambda_{\min}(f_{\emptyset})$.
\end{theorem}

The interlacing family used in this work incorporates a fractional solution $\vx \in \R^m_+$ of the convex programming relaxation in~\eqref{eq:CP} using the following sampling scheme. Suppose we are given vectors $\vv_1, \cdots, \vv_m \in \R^d$ and a fractional solution $\vx \in \R^m_+$ with $\sum_{i=1}^m \vx(i) = k$.
Let $\vu_1 \cdots, \vu_k$ be i.i.d.~\!random vectors sampled from $\{\vv_1, \cdots, \vv_m\}$ with replacement, where $\vv_j$ is sampled with probability $\frac1k \vx(j)$ for all $j \in [m]$ in each round.
Our interlacing family is formed by an $m$-ary tree with depth $k+1$, where the root is labeled by $\emptyset$, and for $0 \leq i \leq k$, each node at the $i$-th level is labeled by some ordered sequence $(s_1, \ldots, s_i) \in [m]^i$.
The polynomial associated with each $(s_1, \ldots, s_i)$ is given by
\begin{align} \label{eq:family}
f_{s_1, \ldots, s_i}(x) := \E\left[ \left. \det \left( x \mI_d - \sum_{j=1}^k \vu_j \vu_j^\top  \right)  \right| \vu_1 = \vv_{s_1}, \cdots, \vu_{i} = \vv_{s_{i}}  \right] .
\end{align}
By construction, the root polynomial $f_{\emptyset}$ is the expected characteristic polynomial. At each level $i$, the polynomial $f_{s_1,...,s_i}$ is the expected characteristic polynomial conditioned on the partial solution $(s_1,...,s_i)$. Finally, each leaf polynomial is the characteristic polynomial of a particular solution. 
It follows from~\cite[Theorem 4.5]{MSS15b} or~\cite[Theorem 2.7]{MSS22} that~\eqref{eq:family} forms an interlacing family. Furthermore, using the algorithm in~\cite[Section 4.1]{MSS22}, each individual polynomial in the interlacing family can be computed efficiently.

\begin{theorem}[\cite{MSS22}] \label{t:family}
    The family of polynomials $\{f_{s_1, \ldots, s_i}(x)\}_{(s_1, \ldots, s_i)}$ in \eqref{eq:family} together with the associated $m$-ary tree structure form an interlacing family. For each fixed $(s_1, \ldots, s_i) \in [m]^i$, the polynomial $f_{s_1, \ldots, s_i}(x)$ can be computed in time $O(d^{\omega+1})$.
\end{theorem}

The following lemma gives a closed-form expression for the root polynomial\footnote{Note that \cite{MSS22} considered the setting where $k \leq d$, but the lemma still holds for $k > d$.}, which is crucial in our analysis.

\begin{lemma}[{\cite[Lemma 4.2]{MSS22}}] \label{l:root-poly}
Suppose we are given $\vv_1, \ldots, \vv_m \in \R^d$ and $\vx \in \R^m_+$, let $\mX = \sum_{i=1}^m \vx(i) \vv_i \vv_i^\top$ and $\lambda_1, \ldots, \lambda_d$ be the eigenvalues of $\mX$. The root polynomial of the interlacing family \eqref{eq:family} is
\begin{align} \label{eq:root}
f_{\emptyset}(x) = x^{d-k} \prod_{i=1}^d \bigg( 1 - \frac{\lambda_i}{k} \partial_x \bigg) x^k.
\end{align}
\end{lemma}

\section{Interlacing Family Based Algorithms for E/D/A-Design} \label{s:algorithms}

We present the interlacing family based algorithms for E/D/A-design and their analyses in this section. 
The algorithms are essentially the same for all E/D/A-design except for a standard normalization preprocessing that simplifies the analysis for E-design. 
Given an optimal fractional solution $\vx$ to \eqref{eq:CP} and let $\mX = \sum_{i=1}^m \vx(i) \vv_i\vv_i^\top$, we normalize the input vectors $\vw_i = \mX^{-\frac12}\vv_i$ for all $i \in [m]$ so that $\sum_{i=1}^m \vx(i) \vw_i\vw_i^\top = \mI_d$. 
If we find a solution $(s_1,...s_k) \in [m]^k$ such that $\lambda_{\min}(\sum_{i=1}^k \vw_{s_i}\vw_{s_i}^\top) \geq \gamma$, then it follows that
\begin{align} \label{eq:minE}
\sum_{i=1}^k \mX^{-1/2}\vv_{s_i} \vv_{s_i}^\top \mX^{-1/2}\succcurlyeq \gamma \mI_d ~~ \implies ~~ \lambda_{\min}\bigg( \sum_{i=1}^k \vv_{s_i} \vv_{s_i}^\top \bigg) \geq \gamma \cdot \lambda_{\min}(\mX),
\end{align}
which implies that $(s_1,...,s_k)$ is an $1/\gamma$-approximate solution for E-design. 

\begin{framed}{

 \noindent \textbf{Algorithm for E/D/A-Design}}
 
 \textbf{Input:} $\vv_1, \ldots, \vv_m \in \R^d$ and $k \geq d$.

 \begin{enumerate}
     \item Solve the convex programming relaxation \eqref{eq:CP} for E/D/A-design to get an optimal solution $\vx \in \R^m_+$ with $\sum_{i=1}^m \vx(i) = k$. Let $\mX = \sum_{i=1}^m \vx(i) \vv_i \vv_i^\top$.

     \item For E-design, let $\vw_i \gets \mX^{-1/2} \vv_i$, for $i \in [m]$, so that $\sum_{i=1}^m \vx(i) \vw_i \vw_i^\top = \mI_d$. Otherwise, let $\vw_i \gets \vv_i$.

     \item Define an interlacing family with respect to $\vx$ and $\vw_i$'s as in~\eqref{eq:family}.

    \item For $i = 1$ to $k$,
    \begin{enumerate}
        \item Let $(s_1, \cdots, s_{i-1})$ be the current internal node and denote $p = f_{s_1, \ldots, s_{i-1}}$ and $p_t = f_{s_1, \ldots, s_{i-1}, t}$ for each $t \in [m]$.
        \item Find a $t \in [m]$ such that 
        \begin{align*}
           \lambda_{\min}(p_{t}) &\geq  \lambda_{\min}(p) \quad \quad \hspace{0.05cm} ~~~\text{for E-design},\\ 
            - \frac{p'_{t}(0)}{p_{t}(0)} & \leq  - \frac{p'(0)}{p(0)} \quad \quad \hspace{0.41cm} \hspace{0.3mm} \text{for A-Design\footnotemark}, \\
            (-1)^d p_{t}(0) & \geq  (-1)^d p(0) \quad \quad \text{for D-Design}.
        \end{align*}
        Set $s_i \gets t$.
    \end{enumerate}

    \item Return $(s_1, \cdots, s_k)$ as the final solution.
 \end{enumerate}
 \end{framed}

 \footnotetext{We make the following assumption to deal with the dividing by zero issue implicitly. For any polynomial $p$ and $q$ with $p(0) = q(0) = 0$, we assume $-p'(0)/p(0) = -q'(0)/q(0) = \infty$, where $\infty$ denotes a special value that is strictly greater than any finite real number. \label{fn:as}}

To analyze the runtime of the rounding process, we note that each of the polynomials $f_{s_1, \ldots, s_i}$ can be computed in $O(d^{\omega+1})$ time by \autoref{t:family}.  And the minimum root of $f_{s_1, \ldots, s_i}$ can be computed up to an $\eps$ additive error in time $O(d^2 \log \eps^{-1})$~\cite[Section 4.1]{MSS22} where setting $\eps = 1/\poly(d)$ suffices for our purpose. The quantities $p(0)$ and $p'(0)$ are simply the constant and linear coefficients of $p$, which can be obtained in constant time. 
Therefore, the whole rounding process (Step 4) runs in $O(kmd^{\omega+1})$ time.

In the remaining of this section, we analyze the algorithms for E/D/A-design in \Cref{ss:E}, \Cref{ss:D}, and \Cref{ss:A} respectively.
The analyses in all three settings follow the same framework with two parts.
\begin{itemize}
    \item First, we prove that the interlacing family satisfies certain ``sandwiching'' property for the corresponding objectives, which implies that Step 4 of the algorithm terminates successfully and returns a solution $(s_1, \ldots, s_k)$ with objective value no worse than the root polynomial $f_{\emptyset}$.
    Note that we actually only need one side of the ``sandwiching'' property for the analysis of the algorithm.

    \item Second, we relate the objective value of the root polynomial to the optimal value of the convex programming relaxation \eqref{eq:CP}.
\end{itemize}

\subsection{Analysis for E-Design} \label{ss:E}

We analyze the algorithm for E-design and prove \autoref{t:E-design} in this subsection. 
For E-design, the ``sandwiching'' property for $\lambda_{\min}$ follows directly from \autoref{t:interlacing} as \eqref{eq:family} is an interlacing family.
Thus, we can always find the desired $s_i=t$ in Step (4b) in each iteration.
Therefore, the algorithm will return a solution $(s_1, \ldots, s_k)$ such that
\begin{align} \label{eq:lambda_min}
\lambda_{\min}(f_{s_1, \ldots, s_k}) = \lambda_{\min}\bigg( \sum_{i=1}^k \vw_{s_i} \vw_{s_i}^\top \bigg) \geq \lambda_{\min}(f_{\emptyset}).
\end{align}
To relate the objective of the root polynomial to the optimal value of the convex relaxation, it suffices to provide a lower bound on $\lambda_{\min}(f_{\emptyset})$.
Since we have normalized the inputs so that $\sum_{i=1}^m \vx(i) \vw_i \vw_i^\top = \mI_d$, \Cref{l:root-poly} implies that
\[
f_\emptyset(x) =  x^{d-k} \bigg( 1 - \frac{1}{k} \partial_x \bigg)^d x^k.
\]
To proceed, we consider an alternative closed-form for the root polynomial $f_{\emptyset}$. This form has been shown in~\cite[Section 2 and 4]{MSS22} and we provide a proof in \Cref{a:omitted} for completeness.
\begin{lemma}[\cite{MSS22}] \label{l:alter-root}
    \[
        f_{\emptyset}(x) = x^{d-k} \bigg( 1 - \frac{1}{k} \partial_x \bigg)^d x^k = \bigg( 1 - \frac1k \partial_x \bigg)^k x^d.
    \]
\end{lemma}

To lower bound the minimum root of a polynomial $p$, we will lower bound a smoother soft-min function
\[
\alphamin(p) := \lambda_{\min}(p + \alpha p')
\]
where $\alpha > 0$ is a parameter.
We note that $\alphamin(p) < \lambda_{\min}(p)$\footnote{To see this, we observe that $p(y) + \alpha p'(y) = p(y) (1 + \alpha p'(y)/p(y))$ and $p'(y)/p(y)$ is negative and monotone for $y < \lambda_{\min}(p)$. Thus, there is a unique $y < \lambda_{\min}(p)$ such that $p(y) + \alpha p'(y) = 0$.}.

We will use the following two lemmas from~\cite{MSS22} to reason about $\alphamin$. The first one quantifies the increase of $\alphamin$ under the linear operator $(1-\lambda \partial_x)$.
\begin{lemma}[{\cite[Lemma 4.3]{MSS22}}] \label{l:alpha-min}
If $p(x)$ is a real-rooted polynomial and $\lambda > 0$, then $(1-\lambda \partial_x) p$ is also real-rooted and it holds for any $\alpha > 0$ that
\[
\alphamin((1-\lambda\partial_x) p) \geq \alphamin(p) + \frac{1}{1/\lambda + 1/\alpha}.
\]
\end{lemma}
The second one provides an inequality that is stronger than $\alphamin(p) < \lambda_{\min}(p)$.
\begin{lemma}[{\cite[Claim 5.8]{MSS22}}] \label{l:alpha-lambda}
For every real-rooted polynomial $p(x)$ and $\alpha > 0$,
\[
\alphamin(p) + \alpha \leq \lambda_{\min}(p).
\]
\end{lemma}

Starting with the closed-form of $f_\emptyset(x)$ in \autoref{l:alter-root}, apply \autoref{l:alpha-lambda} to $f_{\emptyset}(x)$ first and then apply \autoref{l:alpha-min} with $\lambda = 1/k$ for $k$ times, we obtain that
\begin{align*} 
\lambda_{\min}(f_{\emptyset}) \geq \alpha + \alphamin(f_{\emptyset})  \geq \alpha + \alphamin(x^d) + \frac{k}{k + 1/\alpha} = -\alpha (d-1) + \frac{k}{k + 1/\alpha},
\end{align*}
where we used the fact that $\alphamin(x^d) = \lambda_{\min}(x^d + \alpha d x^{d-1})= - \alpha d$.

We optimize $\alpha$ %over the RHS of \eqref{eq:lambda-min} 
and then take $\alpha = \frac{\sqrt{(d-1)k} - (d-1)}{(d-1)k}$ (which is positive for $k \geq d$) to obtain that
\begin{align*}
    \lambda_{\min}(f_\emptyset) & \geq  -\alpha (d-1) + \frac{\alpha k}{\alpha k + 1} \\
    & = - \frac{\sqrt{(d-1)k} - (d-1)}{k} + \frac{(\sqrt{(d-1)k} - (d-1))/(d-1)}{(\sqrt{(d-1)k} - (d-1))/(d-1) + 1} \\
    & = \frac{d-1}{k} - 2  \sqrt{\frac{d-1}{k}} + 1 \\
    & = \bigg( 1 -  \sqrt{\frac{d-1}{k}} \bigg)^2. 
\end{align*}

Using this lower bound with \eqref{eq:minE} and \eqref{eq:lambda_min}, we have finished the proof of \Cref{t:E-design}.

\subsection{Analysis for D-Design} \label{ss:D}

We will analyze the algorithm for D-design and prove \Cref{t:D-design} in this subsection.
We first show the ``sandwiching'' property for the D-design objective.

\begin{claim} \label{l:trickle-D}
Let $p_1, \ldots, p_m$ be real-rooted polynomials of degree $d$, and let $p$ be a convex combination of the $m$ polynomials, such that $p = \sum_{i=1}^m \mu_i p_i$ for $\sum_{i=1}^m \mu_i = 1$ and $\mu_i \geq 0$ for all $i \in [m]$. Then
\[
\min_{t \in [m]} (-1)^d \cdot p_t(0) \leq (-1)^d \cdot p(0) \leq \max_{t \in [m]} (-1)^d \cdot p_t(0).
\]
\end{claim}
\begin{proof}
Since $p(0)$ and $p_t(0)$ are the constant coefficients of the polynomials, the lemma follows directly from the assumption that $p$ is a convex combination of $p_i$'s.
\end{proof}

Therefore, the algorithm for D-design will successfully return a leaf $(s_1, \ldots, s_k)$ with
\[
\det\left( \sum_{i=1}^k \vv_{s_i} \vv_{s_i}^\top \right) = (-1)^d f_{s_1, \ldots, s_k}(0) \geq (-1)^d f_{\emptyset}(0).
\]

Then, we relate the objective value of the root polynomial to the one of the fractional solution.

\begin{lemma} \label{l:obj-D}
Let $\mX = \sum_{i=1}^m \vx(i) \vv_i \vv_i^\top$. Then
\[
(-1)^d f_{\emptyset}(0) = \frac{k!}{(k-d)! k^d} \det(\mX). 
\]
\end{lemma}

\begin{proof}
Let $\lambda_1, \ldots, \lambda_d$ be the eigenvalues of $\mX$.
It follows from~\eqref{eq:root} that the constant coefficient of $f_{\emptyset}$ is
\begin{align} \label{eq:constant}
f_{\emptyset}(0) = x^{d-k} \left( \prod_{i=1}^d \bigg( - \frac{\lambda_i}{k} \partial_x \bigg) x^k \right) = (-1)^d \frac{k!}{(k-d)! k^d} \prod_{i=1}^d \lambda_i = (-1)^d \frac{k!}{(k-d)! k^d} \det(\mX).
\end{align}
\end{proof}

Therefore, we conclude that the returned solution $(s_1, \ldots, s_k)$ is a $k \cdot ((k-d)!/k!)^{1/d}$-approximate solution for the D-design problem, proving \Cref{t:D-design}. 

\begin{remark*} 
Since $(-1)^df_{s_1,...s_i}(0) = \E[\det(\sum_{j=1}^k \vu_j \vu_j^\top)|\vu_1=\vv_{s_1},...\vu_i=\vv_{s_i}]$, our algorithm is exactly the method of conditional expectations, where we have used \autoref{t:family} and \autoref{l:root-poly} to compute the conditional expectations explicitly.
\end{remark*}

\subsection{Analysis for A-Design} \label{ss:A}

We will analyze the algorithm for A-design and prove \Cref{t:A-design} in this subsection.
Again, we first prove the ``sandwiching'' property for the A-design objective. Here we need an additional assumption that the polynomials have minimum roots of at least 0. Note that each polynomial in the family is a convex combination of characteristic polynomials of PSD matrices. This means they all have the same sign for all non-positive $x$, which implies their convex combination has minimum root at least $0$.

\begin{lemma} \label{l:trickle-A}
Let $p_1, \ldots, p_m$ be real-rooted monic polynomials of the same degree and with minimum roots at least~$0$. Let $p$ be a convex combination of the $m$ polynomials, such that $p = \sum_{i=1}^m \mu_i p_i$ for $\sum_{i=1}^m \mu_i = 1$ and $\mu_i \geq 0$ for all $i \in [m]$. Then
\[
\min_{t \in [m]} - \frac{p'_t(0)}{p_t(0)} \leq - \frac{p'(0)}{p(0)} \leq \max_{t \in [m]} - \frac{p'_t(0)}{p_t(0)}.
\]
\end{lemma}
\begin{proof}
If $p(0) = 0$ then $p_i(0)=0$ for all $i \in [m]$ and the lemma holds trivially (recall the assumption in \cref{fn:as}). Thus, we assume $p(0) \neq 0$.
Since each monic $p_i$ has the same degree and the minimum root is at least 0, $p_i(0)$ has the same sign (could be 0) for all $i \in [m]$.
Similarly, $p'_i(0)$ has the same sign (which is actually the opposite sign of $p_i(0)$) for all $i \in [m]$. Hence, all $-p'_i(0)$ and $p_i(0)$ have the same sign. Therefore,
\[
\min_{t \in [m]} - \frac{p'_t(0)}{p_t(0)} \leq - \frac{p'(0)}{p(0)} = - \frac{\sum_{i=1}^m \mu_i p'_i(0)}{\sum_{i=1}^m \mu_i p_i(0)} \leq \max_{t \in [m]} - \frac{p'_t(0)}{p_t(0)}. \qedhere
\]
\end{proof}

\autoref{l:trickle-A} implies that the for loop in the algorithm for A-design will terminate successfully.
The algorithm returns a leaf $(s_1, \ldots, s_k)$ such that
\[
\tr\left( \left( \sum_{i=1}^k \vv_{s_i} \vv_{s_i}^\top \right)^{-1} \right) = - \frac{f'_{s_1, \ldots, s_k}(0)}{f_{s_1, \ldots, s_k}(0)} \leq - \frac{f'_{\emptyset}(0)}{f_{\emptyset}(0)}.
\]

Finally, we relate the objective value of the root polynomial to the one of the fractional solution.

\begin{lemma} \label{l:obj-A}
Let $\mX = \sum_{i=1}^m \vx(i) \vv_i \vv_i^\top$. Then
\[
- \frac{f'_{\emptyset}(0)}{f_{\emptyset}(0)} \leq \frac{k}{k-d+1} \tr(\mX^{-1}).
\]
\end{lemma}
\begin{proof}
Let $\lambda_1, \ldots, \lambda_d$ be the eigenvalues of $\mX$.
Note that $f_{\emptyset}(0)$ is the constant coefficient of $f_{\emptyset}$ given by \eqref{eq:constant}.
Meanwhile, $f'_{\emptyset}(0)$ is the coefficient of the linear term of $f_{\emptyset}$. Thus, according to \eqref{eq:root}, we have
\begin{align} \label{eq:degree1}
f'_{\emptyset}(0) = \left. x^{d-k} \left(\sum_{i=1}^d \prod_{j \neq i} \bigg( - \frac{\lambda_j}{k} \partial_x \bigg) x^k \right) \right|_{x=0} = (-1)^{d-1} \sum_{i=1}^d \frac{k!}{(k-d+1)! k^{d-1}} \prod_{j \neq i} \lambda_j.
\end{align}
The lemma follows directly by combining~\eqref{eq:constant} and~\eqref{eq:degree1}. \qedhere
\end{proof}

Therefore, we conclude that the returned solution $(s_1, \ldots, s_k)$ is a $k/(k-d+1)$-approx solution for the A-design problem, proving \Cref{t:A-design}.

\section*{Concluding Remarks and Open Questions}

The main conceptual contribution of this paper is to apply the interlacing polynomial method effectively in a new domain, and the key technical contribution is to relate the polynomials to the convex programming relaxations for the experimental design problems.
We find it very nice that this provides a unifying framework to obtain sharp bounds in these problems.

Two settings have been studied in experimental design.
One is the ``with repetitions'' setting where the solution $S$ is allowed to be a multiset, which is the setting in this paper.
Another is the ``without repetitions'' setting where each vector can only be chosen at most once, which is a more general setting as one can simply add multiple copies of each vector.
We do not yet know how to apply the interlacing polynomial method to the without repetitions setting as it is not clear how to define the probability distributions and expected polynomials appropriately.
We leave this as an open question for further study.

We also note that \autoref{t:E-design} implies a $d^2$-approximation algorithm for E-design when $k=d$, but the best known integrality gap of the relaxation~\eqref{eq:CP} in this case is only $d$. 
It is an interesting open question to close this $d$ vs $d^2$ gap, as it may lead to improvements for the restricted invertibility problem as well.

\section*{Acknowledgements}

We thank anonymous reviewers for helpful suggestions.

\bibliographystyle{alpha}
\bibliography{references}

\begin{thebibliography}{AZLSW21}

\bibitem[AZLSW17]{AZLSW17c}
Zeyuan Allen-Zhu, Yuanzhi Li, Aarti Singh, and Yining Wang.
\newblock Near-optimal design of experiments via regret minimization.
\newblock In {\em Proceedings of the 34th International Conference on Machine Learning}, volume~70, pages 126--135, 2017.

\bibitem[AZLSW21]{AZLSW21}
Zeyuan Allen-Zhu, Yuanzhi Li, Aarti Singh, and Yining Wang.
\newblock Near-optimal discrete optimization for experimental design: a regret minimization approach.
\newblock {\em Math. Program.}, 186(1-2, Ser. A):439--478, 2021.

\bibitem[cMI09]{CM09}
Ali \c{C}ivril and Malik Magdon-Ismail.
\newblock On selecting a maximum volume sub-matrix of a matrix and related problems.
\newblock {\em Theoret. Comput. Sci.}, 410(47-49):4801--4811, 2009.

\bibitem[LWZ23]{LWZ23}
Lap~Chi Lau, Robert Wang, and Hong Zhou.
\newblock Experimental design for any p-norm.
\newblock In {\em Approximation, Randomization, and Combinatorial Optimization. Algorithms and Techniques, {APPROX/RANDOM} 2023}, volume 275 of {\em LIPIcs}, pages 4:1--4:21, 2023.

\bibitem[LZ22a]{LZ22b}
Lap~Chi Lau and Hong Zhou.
\newblock A local search framework for experimental design.
\newblock {\em SIAM J. Comput.}, 51(4):900--951, 2022.

\bibitem[LZ22b]{LZ22a}
Lap~Chi Lau and Hong Zhou.
\newblock A spectral approach to network design.
\newblock {\em SIAM J. Comput.}, 51(4):1018--1064, 2022.

\bibitem[MS17]{MS17}
Zelda Mariet and Suvrit Sra.
\newblock Elementary symmetric polynomials for optimal experimental design.
\newblock In {\em Proceedings of the 31st International Conference on Neural Information Processing Systems}, page 2136–2145, 2017.

\bibitem[MSS15a]{MSS15a}
Adam~W. Marcus, Daniel~A. Spielman, and Nikhil Srivastava.
\newblock Interlacing families {I}: {B}ipartite {R}amanujan graphs of all degrees.
\newblock {\em Ann. of Math. (2)}, 182(1):307--325, 2015.

\bibitem[MSS15b]{MSS15b}
Adam~W. Marcus, Daniel~A. Spielman, and Nikhil Srivastava.
\newblock Interlacing families {II}: {M}ixed characteristic polynomials and the {K}adison-{S}inger problem.
\newblock {\em Ann. of Math. (2)}, 182(1):327--350, 2015.

\bibitem[MSS22]{MSS22}
Adam~W. Marcus, Daniel~A. Spielman, and Nikhil Srivastava.
\newblock Interlacing families {III}: {S}harper restricted invertibility estimates.
\newblock {\em Israel J. Math.}, 247(2):519--546, 2022.

\bibitem[NST22]{NST22}
Aleksandar Nikolov, Mohit Singh, and Uthaipon Tantipongpipat.
\newblock Proportional volume sampling and approximation algorithms for {$A$}-optimal design.
\newblock {\em Math. Oper. Res.}, 47(2):847--877, 2022.

\bibitem[Puk06]{Puk06}
Friedrich Pukelsheim.
\newblock {\em Optimal design of experiments}, volume~50 of {\em Classics in Applied Mathematics}.
\newblock Society for Industrial and Applied Mathematics (SIAM), Philadelphia, PA, 2006.
\newblock Reprint of the 1993 original.

\bibitem[SEFM15]{SEFM15}
Marco~Di Summa, Friedrich Eisenbrand, Yuri Faenza, and Carsten Moldenhauer.
\newblock On largest volume simplices and sub-determinants.
\newblock In {\em Proceedings of the Twenty-Sixth Annual ACM-SIAM Symposium on Discrete Algorithms}, page 315–323, 2015.

\bibitem[SX20]{SX20}
Mohit Singh and Weijun Xie.
\newblock Approximation algorithms for {$D$}-optimal design.
\newblock {\em Math. Oper. Res.}, 45(4):1512--1534, 2020.

\bibitem[XX21]{XX21}
Jiaxin Xie and Zhiqiang Xu.
\newblock Subset selection for matrices with fixed blocks.
\newblock {\em Israel J. Math.}, 245(1):1--26, 2021.

\end{thebibliography}

\appendix

\section{Experimental Design With Generalized Ratio Objective} \label{s:ratio}

Mariet and Sra~\cite{MS17} proposed an experimental design objective based on elementary symmetric functions that generalizes both A-design and D-design objectives. Later, Nikolov et al.~\cite{NST22} further extended it to a generalized ratio objective.
Let $0 \leq l' < l \leq d$, the objective in~\cite{NST22} is defined as
\begin{align} \label{eq:ratio}
\left( \frac{E_{l'}\big( \sum_{i \in S} \vv_i \vv_i^\top \big)}{E_l\big( \sum_{i \in S} \vv_i \vv_i^\top \big)} \right)^{\frac{1}{l-l'}},
\end{align}
where $E_l(\mM) := \sum_{S \subseteq [d]: |S| = l} \prod_{i \in S} \lambda_i(\mM)$ is the $l$-th elementary symmetric polynomial of the eigenvalues of the symmetric matrix $\mM$.
Note that the characteristic polynomial of $\mM$ can be written as 
\[
\det(x I - \mM) = \sum_{l=0}^d (-1)^l E_l(\mM) x^{d-l}.
\]
To see that the generalized ratio objective generalizes both A-design and D-design, we note that \eqref{eq:ratio} is $(\det(\sum_{i \in S} \vv_i \vv_i^\top))^{-1/d}$ for $l'=0$ and $l=d$, and it is $\tr((\sum_{i \in S} \vv_i \vv_i^\top)^{-1})$ for $l' = d-1$ and $l=d$.

\begin{remark*}
There are other ways to generalize the objectives. For example, the $p$-norm objective in~\cite{LWZ23} captures all of E/D/A-design objectives as special cases. However, we do not know how to apply the method of interlacing polynomials to the $p$-norm objective.
\end{remark*}

Nikolov, Singh, and Tantipongpipat~\cite{NST22} showed that the experimental design problem with the generalized ratio objective can also be captured by a convex programming relaxation.

\begin{equation} \label{eq:ratio-CP}
    \begin{aligned}
        & \underset{\vx \in \R^n}{\rm minimize} & & \left( \frac{E_{l'}(\sum_{i=1}^n \vx(i) \vv_i \vv_i^\top)}{E_l(\sum_{i=1}^n \vx(i) \vv_i \vv_i^\top)} \right)^{\frac{1}{l - l'}} \\
        & \text{\rm subject to} & & \sum_{i=1}^n \vx(i) \leq k, \\
        & & & \vx(i) \geq 0, \quad ~\text{ for } i \in [n].
    \end{aligned} 
\end{equation}

We will use the general framework in \autoref{s:algorithms} to provide an algorithm for the experimental design problem with the generalized ratio objective.
As in the algorithms for D/A-design, the interlacing family \eqref{eq:family} is defined with respect to the unnormalized input vectors $\vv_i$'s and the optimal fractional solution $\vx$ of \eqref{eq:ratio-CP}.
For $0 \leq l' < l \leq d$ and a leaf $(s_1, \ldots, s_k)$, the generalized ratio objective of the solution $(s_1, \ldots, s_k)$ (ignoring the exponent $1/(l-l')$) is
    \[
        (-1)^{l' -l} \frac{(d-l)! \cdot f_{s_1, \ldots, s_k}^{(d-l')}(0)}{(d-l')! \cdot f_{s_1, \ldots, s_k}^{(d-l)}(0)}.
    \]

To see this, we note that $f_{s_1, \ldots, s_k}^{(d-l)}(0)/(d-l)!$ is the coefficient of $x^{d-l}$ in the characteristic polynomial and
\[
\det\left( x \mI - \sum_{j = 1}^k \vv_{s_j} \vv_{s_j}^\top \right) = \sum_{l=0}^d (-1)^l E_l\bigg( \sum_{j=1}^k \vv_{s_j} \vv_{s_j}^\top \bigg) x^{d-l}.
\]
Therefore,
\[
E_l\bigg( \sum_{j=1}^k \vv_{s_j} \vv_{s_j}^\top \bigg) = (-1)^l \frac{f_{s_1, \ldots, s_k}^{(d-l)}(0)}{(d-l)!} \qquad \text{and} \qquad \frac{E_{l'}\big( \sum_{j=1}^k \vv_{s_j} \vv_{s_j}^\top \big)}{E_l\big( \sum_{j=1}^k \vv_{s_j} \vv_{s_j}^\top \big)} = (-1)^{l' -l} \frac{(d-l)! \cdot f_{s_1, \ldots, s_k}^{(d-l')}(0)}{(d-l')! \cdot f_{s_1, \ldots, s_k}^{(d-l)}(0)}.
\]

The generalized ratio objective can be naturally extended to all the internal polynomials of the interlacing family in the same way as the E/D/A-design objectives.
Now, we are ready to state the algorithm.

\begin{minipage}{\textwidth}
 \begin{framed}{
 
 \noindent \textbf{Algorithm for The Generalized Ratio Objective}}
 
 \textbf{Input:} $\vv_1, \cdots, \vv_m \in \R^d$ and $k \geq d$.
\vspace{-5pt}
 \begin{enumerate}
     \item Solve the convex programming relaxation \eqref{eq:ratio-CP} and get an optimal solution $\vx \in \R^m_+$ with $\sum_{i=1}^m \vx(i) = k$. 

     \item Let the interlacing family in~\eqref{eq:family} be defined with respect to $\vx$ and the normalized $\vv_i$'s.
   
    \item For $i = 1$ to $k$,
    \begin{enumerate}
        \item Let $(s_1, \ldots, s_{i-1})$ be the current internal node and denote $p = f_{s_1, \ldots, s_{i-1}}$ and $p_t = f_{s_1, \cdots, s_{i-1}, t}$ for each $t \in [m]$.
        \item Find a $t \in [m]$ such that\footnotemark 
        \begin{align*}
            (-1)^{l'-l} \cdot \frac{(d-l)! \cdot p_{t}^{(d-l')}(0)}{(d-l')! \cdot p_{t}^{(d-l)}(0)} \leq (-1)^{l'-l} \cdot \frac{(d-l)! \cdot p^{(d-l')}(0)}{(d-l')! \cdot p^{(d-l)}(0)}.
        \end{align*}
        Set $s_i \gets t$.
    \end{enumerate}

    \item Return $(s_1, \ldots, s_k)$ as the final solution.
 \end{enumerate}
 \vspace{-5pt}
 \end{framed}
\end{minipage}

\footnotetext{We can deal with the dividing by zero issue in a similar way as in \cref{fn:as} for the A-design objective.}

The generalized ratio objective also satisfies a ``sandwiching'' property with respect to the interlacing family.

\begin{lemma} \label{l:trickle-R}
Let $p_1, \ldots, p_m$ be real-rooted monic polynomials of degree $d$ and the smallest root of each polynomial is at least 0. Let $p$ be a convex combination of the $m$ polynomials, such that $p = \sum_{i=1}^m \mu_i p_i$ for $\sum_{i=1}^m \mu_i = 1$ and $\mu_i \geq 0$ for all $i \in [m]$. Then
\[
\min_{t \in [m]} (-1)^{l'-l} \cdot \frac{(d-l)! \cdot p_t^{(d-l')}(0)}{(d-l')! \cdot p_t^{(d-l)}(0)} \leq (-1)^{l'-l} \cdot \frac{(d-l)! \cdot p^{(d-l')}(0)}{(d-l')! \cdot p^{(d-l)}(0)} \leq \max_{t \in [m]} (-1)^{l'-l} \cdot \frac{(d-l)! \cdot p_t^{(d-l')}(0)}{(d-l')! \cdot p_t^{(d-l)}(0)}.
\]
\end{lemma}
\begin{proof}
We can check that if $p^{(d-l)}(0) = 0$ then the lemma holds trivially. 
Thus, we assume $p^{(d-l)}(0) \neq 0$.
Since all monic $p_i$'s have the same degree and the smallest roots are at least 0, all $p^{(d-l')}_i(0)$ have the same sign $(-1)^{l'}$ and all $p^{(d-l)}_i(0)$ have the same sign $(-1)^l$. Therefore,
\[
(-1)^{l'-l} \cdot \frac{(d-l)! \cdot p^{(d-l')}(0)}{(d-l')! \cdot p^{(d-l)}(0)} = (-1)^{l'-l} \cdot \frac{(d-l)! \cdot \sum_{i=1}^m \mu_i p_i^{(d-l')}(0) }{(d-l')! \cdot \sum_{i=1}^m \mu_i p_i^{(d-l)}(0) } \geq \min_{t \in [m]} (-1)^{l'-l} \cdot \frac{(d-l)! \cdot p_t^{(d-l')}(0)}{(d-l')! \cdot p_t^{(d-l)}(0)}.
\]
The second inequality in the lemma follows similarly.
\end{proof}

Therefore, the algorithm returns a leaf $(s_1, \ldots, s_k)$ such that
\[
\frac{E_{l'}(\sum_{i =1}^k \vv_{s_i} \vv_{s_i}^\top)}{E_l(\sum_{i =1}^k \vv_{s_i} \vv_{s_i}^\top)} = (-1)^{l'-l} \cdot \frac{(d-l)! \cdot f^{(d-l')}_{s_1, \ldots, s_k}(0)}{(d-l')! \cdot f^{(d-l)}_{s_1, \ldots, s_k}(0)} \leq (-1)^{l'-l} \cdot \frac{(d-l)! \cdot f^{(d-l')}_{\emptyset}(0)}{(d-l')! \cdot f^{(d-l)}_{\emptyset}(0)}.
\]

Finally, we relate the objective value of the root polynomial to the one of the fractional solution.

\begin{lemma} \label{l:obj-R}
Let $\mX = \sum_{i=1}^m \vx(i) \vv_i \vv_i^\top$. Then
\[
(-1)^{l'-l} \cdot \frac{(d-l)! \cdot f^{(d-l')}_{\emptyset}(0)}{(d-l')! \cdot f^{(d-l)}_{\emptyset}(0)} = \frac{(k-l)! \cdot k^l }{(k-l')! \cdot k^{l'}} \cdot \frac{E_{l'}(\mX)}{E_l(\mX)} \leq \Big( \frac{k}{k-l+1} \Big)^{l-l'} \frac{E_{l'}(\mX)}{E_l(\mX)}.
\]
When $k=l$, the bound can be improved to
\[
(-1)^{l'-l} \cdot \frac{(d-l)! \cdot f^{(d-l')}_{\emptyset}(0)}{(d-l')! \cdot f^{(d-l)}_{\emptyset}(0)} = \frac{l^{l-l'} }{(l-l')!} \cdot \frac{E_{l'}(\mX)}{E_l(\mX)} \leq \Big( \frac{e l}{l-l'} \Big)^{l-l'} \frac{E_{l'}(\mX)}{E_l(\mX)}.
\]
\end{lemma}
\begin{proof}
Let $\lambda_1, \ldots, \lambda_d$ be the eigenvalues of $\mX$.
We note that $f_{\emptyset}^{(d-l)}(0)/(d-l)!$ is the coefficient of $x^{d-l}$ in the polynomial $f_{\emptyset}(x)$. According to~\eqref{eq:root}, it is also the coefficient of $x^{k-l}$ in the polynomial of 
\[
\prod_{i=1}^d \Big( 1 - \frac{\lambda_i}{k} \partial_x \Big) x^k = \sum_{S \subseteq [m]} \prod_{i \in S} \Big(-\frac{\lambda_i}{k} \partial_x \Big) x^k.
\]
Therefore, it follows that
\[
\frac{f^{(d-l)}_{\emptyset}(0)}{(d-l)!} = (-1)^l \frac{k!}{(k-l)! k^l} \sum_{S \subseteq [m]: |S| = l} \prod_{i \in S} \lambda_i = (-1)^l \frac{k!}{(k-l)! k^l} E_l(\mX).
\]
Similarly, we have
\[
\frac{f^{(d-l')}_{\emptyset}(0)}{(d-l')!} = (-1)^{l'} \frac{k!}{(k-l')! k^{l'}} \sum_{S \subseteq [m]: |S| = l'} \prod_{i \in S} \lambda_i = (-1)^{l'} \frac{k!}{(k-l')! k^{l'}} E_{l'}(\mX).
\]
Therefore, the first part of the lemma follows by observing that
\[
\frac{(k-l)!\cdot k^l}{(k-l')! \cdot k^{l'}} = \frac{k^{l-l'}}{(k-l')(k-l'-1) \cdots (k-l+1)} \leq \Big( \frac{k}{k-l+1} \Big)^{l-l'}.
\]

When $k=l$, the second part of the lemma follows by Stirling's approximation that
\[
\frac{l^{l-l'}}{(l-l')!} \leq \frac{l^{l-l'}}{(l-l')^{l-l'} e^{-(l-l')}} = \Big( \frac{e l}{l-l'} \Big)^{l-l'}. \qedhere
\]
\end{proof}

Combining everything together, we conclude with the following theorem that generalizes both \Cref{t:D-design} and \Cref{t:A-design}.

\begin{theorem}
For any $0 \leq l' < l \leq d$, there is a $k \cdot ((k-l)! / k!)^{1/(l-l')}$-approximation rounding algorithm based on the method of interlacing polynomials for the experimental design problem with the generalized ratio objective, and the rounding process runs in $O(kmd^{\omega+1})$ time.
\end{theorem}

\section{Omitted Proof} \label{a:omitted}

\begin{proof}[Proof of \Cref{l:alter-root}]
    First, we note that
    \begin{align*}
        x^{d-k} \left( 1 - \frac1k \partial_x \right)^d x^k & = x^{d-k} \sum_{i=0}^d (-1)^i \binom{d}{i} \frac{1}{k^i} (\partial_x)^i x^k \\
        & = x^{d-k} \sum_{i=0}^d (-1)^i \binom{d}{i} \frac{k!}{k^i (k-i)!} x^{k-i} \\
        & = \sum_{i=0}^d (-1)^i \binom{d}{i} \binom{k}{i} \frac{i!}{k^i} x^{d-i}.
    \end{align*}

    Similarly, we have
    \begin{align*}
        \left( 1 - \frac1k \partial_x \right)^k x^d & = \sum_{i=0}^d (-1)^i \binom{k}{i} \frac{1}{k^i} (\partial_x)^i x^d \\
        & = \sum_{i=0}^d (-1)^i \binom{k}{i} \frac{d!}{k^i (d-i)!} x^{d-i} \\
        & = \sum_{i=0}^d (-1)^i \binom{d}{i} \binom{k}{i} \frac{i!}{k^i} x^{d-i}. \qedhere
    \end{align*}
\end{proof}

\end{document}